%% file: main.tex
\documentclass[11pt]{article}

\usepackage[T1]{fontenc}
\usepackage{lmodern}
\usepackage{microtype}
\usepackage{flafter}
\usepackage{placeins}
\usepackage[a4paper,margin=28mm]{geometry}
\usepackage{amsmath,amssymb,amsthm}
\usepackage{booktabs,tabularx,array}
\usepackage{tikz}
\usetikzlibrary{arrows.meta}
\usepackage{xcolor}
\usepackage{hyperref}

\definecolor{ink}{HTML}{18232D}
\definecolor{muted}{HTML}{66737E}
\definecolor{laneblue}{HTML}{176B87}
\definecolor{lanepale}{HTML}{DDEFF4}
\definecolor{caprust}{HTML}{B45134}
\definecolor{cappale}{HTML}{F6E4DE}
\definecolor{fixedcell}{HTML}{DCE3E7}
\definecolor{livecell}{HTML}{18232D}
\definecolor{badred}{HTML}{B23A48}

\hypersetup{
  colorlinks=true,
  linkcolor=laneblue,
  citecolor=laneblue,
  urlcolor=laneblue,
  pdftitle={The Ultimate Fate of Life Is Not Shared},
  pdfauthor={Ziyue Gan, Ziran Li, and Jiadong Zhu}
}

\numberwithin{equation}{section}

\newtheorem{lemma}{Lemma}[section]
\newtheorem{theorem}[lemma]{Theorem}
\newtheorem{proposition}[lemma]{Proposition}
\newtheorem{corollary}[lemma]{Corollary}
\newtheorem{fact}[lemma]{Fact}

\newcommand{\Life}{\mathcal{G}}
\newcommand{\Om}{\Omega}
\newcommand{\Z}{\mathbb{Z}}
\newcommand{\cyl}[1]{[#1]}
\newcommand{\shift}{\sigma}
\newcommand{\rot}{\rho}
\newcommand{\refl}{\kappa}

\tikzset{
  proofbox/.style={draw=ink!28,rounded corners=2pt,fill=white,inner sep=7pt,align=center},
  flowarrow/.style={-{Stealth[length=2.2mm]},line width=0.8pt,draw=muted},
  markerframe/.style={draw=ink!45,line width=0.7pt},
  laneframe/.style={draw=laneblue,line width=1.5pt},
  capframe/.style={draw=caprust,densely dashed,line width=1.1pt},
  figurelabel/.style={font=\scriptsize\sffamily,text=ink},
}

\title{The Ultimate Fate of Life Is Not Shared}
\author{Ziyue Gan\thanks{\texttt{ziyuegan@ornisresearch.org}}\\
\small Xi'an Jiaotong University, Xi'an, China\\
\small OrnisResearch, Shanghai, China
\and
Ziran Li\thanks{\texttt{ziranli@ornisresearch.org} (corresponding author).}\\
\small OrnisResearch, Shanghai, China
\and
Jiadong Zhu\thanks{\texttt{zhujiadong2016@163.com}}\\
\small University of Connecticut, Storrs, CT, USA}
\date{}

\begin{document}
\maketitle

\begin{abstract}
The limit set of Conway's Game of Life collects the configurations that can
still appear at arbitrarily late times: those admitting predecessors of every
finite depth.  Answering a question of Salo and T\"orm\"a (ICALP 2022), we
prove that two fixed finite patterns can each occur in the limit set, yet can
never be found together in a single configuration of it, at any relative
position; equivalently, the spatial translation action on the limit set is
not topologically transitive.  Our method is to make a
persistent marker force a periodic lane to grow in sufficiently deep
predecessors, until two perpendicular lanes are forced to intersect and
prescribe incompatible values at a common cell.  The argument
uses two computer-verified local implications.
\end{abstract}

\medskip
\noindent\textbf{Keywords:} Game of Life, cellular automata, limit set,
symbolic dynamics, topological transitivity, ancestral forcing.

\section{Introduction}
\label{sec:intro}

Conway's Game of Life is a two-dimensional cellular automaton whose simple
local rule supports rich local behaviour and universal computation
\cite{gardner1970life,johnston2022conway,rendell2002turing}.  Here we study a
different global question: which configurations remain possible after
arbitrarily many generations, and how freely can their finite patterns
coexist?

Let $\Life$ denote the Life map on $X=\{0,1\}^{\Z^2}$.  Its
\emph{limit set} is
\[
  \Om=\bigcap_{n\geq0}\Life^n(X).
\]
Thus $x\in\Om$ if and only if, for every $n\geq0$, some configuration evolves into $x$
after $n$ steps, with the initial configuration allowed to depend on $n$.
Equivalently, $x$ has arbitrarily long chains of predecessors
\cite{salo-torma2022oracles,salo-torma2025preimages}.
Limit sets are standard objects in cellular automata and symbolic dynamics
\cite{ceccherini2010cellular,kurka2003topological,lind-marcus1995}, and their
languages can be complicated even in one dimension
\cite{hu1987limit,culik1989limit,maass1995sofic}.

The limit set carries both the temporal action of the Life map and the
spatial action of grid translations.  Salo and T\"orm\"a showed that the
temporal action is not topologically transitive
\cite{salo-torma2022oracles} and asked whether the spatial action is transitive
\cite[Question~3]{salo-torma2022oracles}.  In finite-pattern terms, spatial
transitivity asks whether, for every $U$ and $V$ occurring in $\Om$, some
configuration of $\Om$ contains both $U$ and $\shift_vV$ for some
$v\in\Z^2$ \cite{kurka2003topological}.

We answer this question in the negative.

\begin{theorem}
\label{thm:nontransitive}
The spatial action on the limit set of the Game of Life is not topologically
transitive.
\end{theorem}

More precisely, we construct a fixed finite pattern $U$ and its quarter-turn
$V=\rot U$ such that
\[
  \cyl U\cap\Om\neq\varnothing,\qquad
  \cyl V\cap\Om\neq\varnothing,\qquad
  \cyl U\cap\cyl{\shift_vV}\cap\Om=\varnothing
  \quad\text{for every }v\in\Z^2,
\]
where $\cyl P$ denotes the set of configurations agreeing with $P$ on its
specified domain.  The same two patterns therefore exclude coexistence at
every relative position.

The distinction from earlier spatial results is logical rather than merely
quantitative.  Salo and T\"orm\"a proved that $\Om$ is non-sofic
\cite[Theorem~3]{salo-torma2022oracles}.
They also constructed distance-dependent obstructions
\cite[Theorem~7]{salo-torma2022oracles}, ruling out a sublinear gluing radius
in the cardinal directions.  None of these results supplies two
fixed cylinder sets whose shift orbits are disjoint.  Our construction changes
the quantifier order from a new witness at each scale to one witness for every
relative displacement.  Thus $\Om$ fails topological transitivity for a
single fixed pair of nonempty cylinder sets.

The obstacle to obtaining a fixed pair from the earlier construction is the
loss of horizontal width during backward growth
\cite[Lemma~22]{salo-torma2022oracles}.  Iterating that construction to force
a longer lane requires an increasingly wide initial patch.  Our marker
instead carries two pump states whose width is preserved, and the marker is
restored throughout the backward history.  The same fixed seeds can therefore
extend a periodic lane to arbitrary length.

Our main contribution is the fixed separating pair, which answers the open
spatial-transitivity question.  Its proof is organized around
Theorem~\ref{thm:pumping}, which packages finitely checkable local implications
into uniform ancestral signatures.  Proposition~\ref{prop:ancestral-obstruction}
then turns the resulting incompatibility into the spatial obstruction used
here.  Corollary~\ref{cor:quantitative} gives a finite-depth version in which
the excluded displacement range grows linearly with predecessor depth.

The proof mechanism is summarized in Figure~\ref{fig:proof-idea}.  We work in
pairs of Life steps, writing $h=\Life^2$, which has the same limit set as
$\Life$.  The marker $U$ is restored at the same position in every
$h$-predecessor, so its two pump states remain available at every backward
depth.  Repeated local forcing extends a vertical periodic lane in both
directions; the rotated marker similarly extends a horizontal lane.  If both
markers occurred in one configuration of $\Om$, a sufficiently deep common
predecessor would contain both lane segments, and they would prescribe
incompatible values at their intersection.

The proof is partly computer-assisted.  Its two computational ingredients
are local implications for marker forcing and lane pumping, each checked by
finite-state and SAT methods in Appendix~\ref{app:certificates}.
Persistence also uses the published Kynn\"os self-forcing lemma of Salo and
T\"orm\"a \cite[Lemma~21]{salo-torma2022oracles}.  All remaining steps,
including the passage from these finite implications to arbitrary predecessor
depth, are mathematical deductions proved in Section~\ref{sec:proof}.

Section~\ref{sec:prelim} gives the background.  Section~\ref{sec:proof}
proves the general reduction, constructs the Life-specific patterns, and
verifies the separating lemma.  Appendix~\ref{app:certificates} records the
reproducible finite-state and SAT proofs.

\begin{figure}[t]
\centering
\begin{tikzpicture}[font=\small\sffamily]
  \node[proofbox,minimum width=3.35cm,minimum height=2.35cm] (marker) at (0,0) {};
  \draw[fill=ink!9,draw=ink!35] (-0.85,-0.72) rectangle (0.85,0.72);
  \draw[fill=lanepale,draw=laneblue,line width=1.2pt] (-0.18,-1.00) rectangle (0.18,1.00);
  \node at (0,-1.38) {$U$ in the present};

  \node[proofbox,minimum width=3.35cm,minimum height=2.35cm] (ancestor) at (4.6,0) {};
  \draw[fill=ink!9,draw=ink!35] (3.75,-0.72) rectangle (5.45,0.72);
  \draw[fill=lanepale,draw=laneblue,line width=1.2pt] (4.42,-1.00) rectangle (4.78,1.00);
  \draw[laneblue,line width=1.2pt,-{Stealth[length=2mm]}] (4.60,0.98) -- (4.60,1.42);
  \draw[laneblue,line width=1.2pt,-{Stealth[length=2mm]}] (4.60,-0.98) -- (4.60,-1.42);
  \node at (4.6,-1.75) {older ancestors};

  \node[proofbox,minimum width=3.55cm,minimum height=2.35cm] (crossing) at (9.4,0) {};
  \draw[fill=lanepale,draw=laneblue,line width=1.2pt] (9.22,-1.00) rectangle (9.58,1.00);
  \draw[fill=cappale,draw=caprust,line width=1.2pt] (8.40,-0.18) rectangle (10.40,0.18);
  \node[text=badred,font=\bfseries\Large] at (9.4,0) {$\times$};
  \node at (9.4,-1.38) {rotated lanes conflict};

  \draw[flowarrow] (marker.east) -- node[above] {$\Life^{-2}$ in $\Om$} (ancestor.west);
  \draw[flowarrow] (ancestor.east) -- node[above] {grow and rotate} (crossing.west);
\end{tikzpicture}
\caption{The three steps of the proof: the pattern $U$ persists in backward
time, its lane grows at both ends, and the vertical lane and its quarter-turn
prescribe incompatible cell values at every crossing.}
\label{fig:proof-idea}
\end{figure}
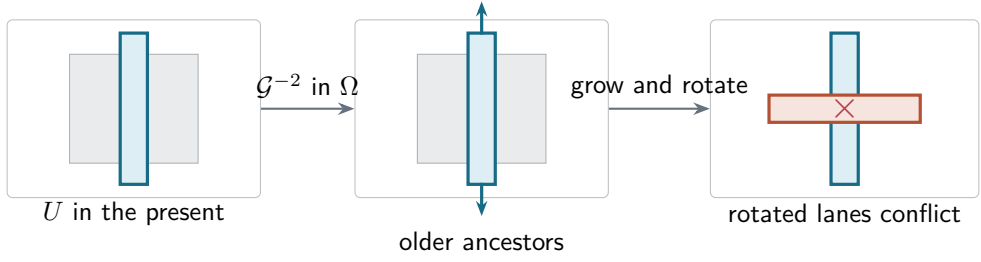

\FloatBarrier

\section{Preliminaries}
\label{sec:prelim}

This section fixes the notation and recalls the symbolic-dynamical background
used in the proof.  We follow the conventions of \cite{salo-torma2022oracles}
and refer to \cite{lind-marcus1995,petersen1983ergodic} for further details.
Throughout, for integers $a\leq b$, we write
$[a,b]=\{a,a+1,\ldots,b\}$.  Let $A$ be a finite
alphabet, let $d\geq1$, and set $Y=A^{\Z^d}$.  A \emph{configuration} is a
point $x\in Y$, and a \emph{pattern} is a map
$P\colon D\to A$ with domain $D=\operatorname{dom}(P)\subset\Z^d$; the
pattern is \emph{finite} if $D$ is finite.  For patterns $P,Q$ we write
$P\sqsubset Q$ when $Q$ extends $P$, regarding a configuration as a pattern
on all of $\Z^d$.  The \emph{cylinder} of a finite pattern $P$ is
\[
  \cyl P=\{x\in Y: P\sqsubset x\}.
\]
Since $A$ is finite, $Y$ is compact, and the finite-pattern cylinders form a
clopen base for its product topology.  The group $\Z^d$ acts on $Y$ by the
\emph{shifts}
\[
  (\shift_v x)(p)=x(p-v),\qquad v,p\in\Z^d.
\]
We call $Y$, together with this action, the \emph{full shift} over $A$ in
dimension $d$.

A \emph{subshift} is a closed, shift-invariant subset $\Lambda\subseteq Y$.
A finite pattern $P$ \emph{occurs in} a set $\Lambda\subseteq Y$ if
$\shift_v\cyl P\cap\Lambda\neq\varnothing$ for some $v\in\Z^d$.  If $\Lambda$
is a subshift, shift invariance makes this equivalent to
$\cyl P\cap\Lambda\neq\varnothing$.  A \emph{cellular automaton} (CA) is a continuous map
$H\colon Y\to Y$ commuting with the shifts.  Its \emph{limit set} is
\[
  \Om_H=\bigcap_{n\geq0}H^n(Y).
\]
The limit set is a subshift.  A \emph{predecessor} of a configuration $x$
is a configuration $y$ with $H(y)=x$.  For a finite pattern $P$, the
following criterion characterizes occurrence in the limit set:
\[
 \cyl P\cap\Om_H\neq\varnothing
 \quad\Longleftrightarrow\quad
 \cyl P\cap H^n(Y)\neq\varnothing\quad\text{for every }n\geq0.
\]
Indeed, the sets on the right are nested compact sets, so if all are nonempty,
their intersection is nonempty.  Shift invariance makes occurrence at any
translate equivalent to occurrence at the specified coordinates.

Let $E\subseteq\Om_H$.  Its \emph{orbit saturation} is
$\bigcup_{v\in\Z^d}\shift_vE$.  The spatial action (the action by
translations) on $\Om_H$ is \emph{topologically transitive} if, for every
pair of nonempty relatively open sets $E,F\subseteq\Om_H$, there is
$v\in\Z^d$ with
$E\cap\shift_vF\neq\varnothing$.  Two patterns are \emph{incompatible} if
they disagree on some cell of their common domain.

A \emph{symmetry of the lattice} is an affine isometry of $\Z^d$; a symmetry
$\gamma$ acts on configurations by $(\gamma x)(p)=x(\gamma^{-1}p)$, and on
patterns accordingly.
A symmetry that commutes with $H$ preserves $\Om_H$ and acts naturally on
patterns and cylinders.  We use the clockwise quarter-turn $\rot$, acting
by $(\rot x)(a,b)=x(-b,a)$, and the reflection $\refl$ in the $x$-axis,
acting by $(\refl x)(a,b)=x(a,-b)$; both are symmetries of the lattice
commuting with Life.

Throughout, $H\colon Y\to Y$ denotes a general cellular automaton on the full
shift $Y$, with limit set $\Om_H$; for the Game of Life we write $\Life$ for
the map, $X=\{0,1\}^{\Z^2}$ for its full shift, and $\Om$ for its limit set.
The Life rule is
\[
  \Life(x)_p=1
  \quad\Longleftrightarrow\quad
  N_x(p)=3\ \text{ or }\ \bigl(x_p=1\text{ and }N_x(p)=2\bigr),
\]
where $N_x(p)$ counts the live cells in the eight-cell Moore neighbourhood
of $p$.  We write $h=\Life^2$.  Since the sets $\Life^n(X)$ are decreasing
and the even subsequence is cofinal, $\Om=\bigcap_{n\geq0}h^n(X)$.

\section{The proof}
\label{sec:proof}

The proof has a general part and a Life-specific part.
Section~\ref{sec:reduction} develops the general reduction and states the one
Life-specific lemma that remains.  Section~\ref{sec:construction} defines the
concrete patterns, and Section~\ref{sec:separating} verifies that lemma and
completes the proof.

\subsection{The reduction}
\label{sec:reduction}

Figure~\ref{fig:reduction-dag} gives the proof dependencies.
Theorem~\ref{thm:pumping} and
Proposition~\ref{prop:ancestral-obstruction} are both proved completely in
this subsection and do not depend on the later Life construction.  The
theorem builds an ancestral
signature; it is combined with the Life-specific finite results stated in
Section~\ref{sec:separating} to prove Lemma~\ref{lem:separating}.  The
proposition is a separate final step: once Lemma~\ref{lem:separating} has
provided the required patterns, the proposition yields
Theorem~\ref{thm:nontransitive}.  The Life-specific verification is deferred
to Section~\ref{sec:separating}.

\begin{figure}[ht]
\centering
\begin{tikzpicture}[font=\small\sffamily,
  every node/.style={proofbox,text width=3.15cm,minimum height=1.05cm}]
  \node (pump) at (-3.7,1.15)
    {Theorem~\ref{thm:pumping}\\ancestral pumping};
  \node (finite) at (-3.7,-1.15)
    {Section~\ref{sec:separating}\\Life-specific finite results};
  \node (sep) at (0.5,0)
    {Lemma~\ref{lem:separating}\\separating pattern};
  \node (prop) at (0.5,-2.05)
    {Proposition~\ref{prop:ancestral-obstruction}\\signature separation};
  \node (main) at (4.7,-0.65)
    {Theorem~\ref{thm:nontransitive}\\spatial non-transitivity};
  \draw[flowarrow] (pump.east) -- (sep.west);
  \draw[flowarrow] (finite.east) -- (sep.west);
  \draw[flowarrow] (sep.east) -- (main.west);
  \draw[flowarrow] (prop.east) -- (main.west);
\end{tikzpicture}
\caption{Proof dependencies for the reduction.  Every arrow $A\to B$ means
that statement or group of statements $A$ is used in the proof of statement
$B$.}
\label{fig:reduction-dag}
\end{figure}
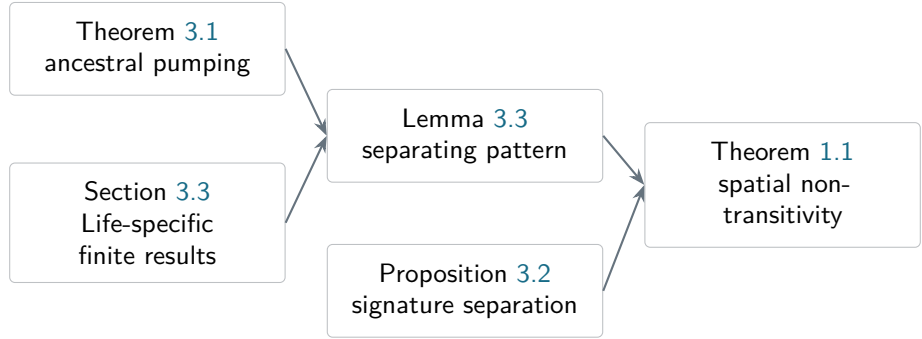

\FloatBarrier

To motivate the ancestral-pumping reduction, consider the backward history
used in the Life construction below.  With $h=\Life^2$, write
$h(x_j)=x_{j-1}$, so that increasing $j$ means moving two Life generations
into the past.  Persistence puts the marker, and hence its two seeds, in
every $x_j$.  The pump uses a seed at the same position in
$x_{j-2},x_{j-1},x_j$ to force an outward translate of that seed in $x_j$.
The first translated seed is therefore present in every layer of depth at
least two, providing the premises for another application of the pump.
Inductively, the $k$th outward translate is present in every layer of depth
at least $2k$.  Each translate advances four rows and keeps the seed's
width fixed.  The two seeds extend the lane in opposite directions.
The ancestral-pumping reduction formalizes this induction and the passage
from finite lane segments to the obstruction at arbitrary separations.

The reduction below holds for any cellular automaton.  Fix a finite alphabet
$A$, a dimension $d\geq1$, a cellular automaton $H\colon Y\to Y$ with
$Y=A^{\Z^d}$, and its limit set $\Om_H$.  We first record that $H|_{\Om_H}$ is
onto.  For $x\in\Om_H$, the compact sets
\[
  H^{-1}(x)\cap H^n(Y),\qquad n\geq0,
\]
are nonempty, because $x\in H^{n+1}(Y)$, and nested; their intersection
therefore contains an $H$-preimage of $x$ in $\Om_H$.

Next we introduce the notion of forcing used in the reduction.  For a
finite pattern $U$, a possibly infinite pattern $B$ is a \emph{uniform
ancestral signature} of $U$ if, for every finite
$D\subseteq\operatorname{dom}(B)$, there exists an integer $N_D\geq0$ such
that
\begin{equation}
 (H|_{\Om_H})^{-n}(\cyl U\cap\Om_H)
 \ \subset\ \cyl{B|_D}\cap\Om_H
 \qquad\text{for every }n\geq N_D.             \label{eq:ancestral-signature}
\end{equation}
In words, for every finite part of $B$, all sufficiently deep predecessors in
the limit set of configurations containing $U$ contain that part at the
prescribed coordinates.

The ancestral-pumping reduction has two parts.  The first shows that a pattern
that persists and successively forces translated copies of finite subpatterns
has a signature.  In the motivating Life application, $H=h=\Life^2$, $\ell=2$,
$q=(0,4)$, $U$ is the persistent marker, and $S_+,S_-$ are its two seeds.
Thus the delay $\ell$ counts applications of $h$, corresponding to four
Life generations for each further outward translate.

\begin{theorem}[Ancestral pumping]
\label{thm:pumping}
Let $q\in\Z^d\setminus\{0\}$, let $\ell\geq1$ be an integer, and let $U$ be
a finite pattern occurring in $\Om_H$.  Let $S_+,S_-$ be finite subpatterns
of $U$, and, for $\epsilon\in\{+1,-1\}$, set $S_\epsilon=S_+$ when
$\epsilon=+1$ and $S_\epsilon=S_-$ when $\epsilon=-1$.  Suppose that every
$y\in Y$ satisfies
\begin{equation}
 \begin{aligned}
 U\sqsubset H(y)&\quad\Longrightarrow\quad U\sqsubset y,\\
 S_\epsilon\sqsubset H^i(y)\quad\text{for every }0\leq i\leq\ell
   &\quad\Longrightarrow\quad \shift_{\epsilon q}S_\epsilon\sqsubset y
   \quad(\epsilon\in\{+1,-1\}).
 \end{aligned}                                          \label{eq:abstract-pump}
\end{equation}
For $m\geq0$, let $W_m$ denote the following formal union of cell-value
prescriptions; its consistency is part of the conclusion:
\begin{equation}
 W_m=U\cup\bigcup_{k=0}^{m}\shift_{kq}S_+
       \cup\bigcup_{k=0}^{m}\shift_{-kq}S_- .          \label{eq:pumped-union}
\end{equation}
Then the following hold.
\begin{enumerate}
\item[(a)] For all integers $m,n\geq0$ and $x,y\in Y$ with $H^n(y)=x$, $U\sqsubset x$,
and $n\geq\ell m$, we have $W_m\sqsubset y$.
\item[(b)] Each $W_m$ is a consistent finite pattern, and
$W=\bigcup_{m\geq0}W_m$ is a consistent pattern with $W\sqsubset z$ for some
$z\in\Om_H$.
\item[(c)] $W$ is a uniform ancestral signature of $U$.
\end{enumerate}
\end{theorem}

\begin{proof}
Let $x=x_0\leftarrow x_1\leftarrow\cdots\leftarrow x_n=y$ be the backward
$H$-chain in $Y$, so that $H(x_j)=x_{j-1}$ for $1\leq j\leq n$.  Repeated
application of the first implication in \eqref{eq:abstract-pump} gives
$U\sqsubset x_j$ for every $0\leq j\leq n$, and hence
$S_+,S_-\sqsubset x_j$ for every such $j$.

We claim that
\begin{equation}
  0\leq j\leq n,\quad j\geq\ell k
  \quad\Longrightarrow\quad
  \shift_{\epsilon kq}S_\epsilon\sqsubset x_j
  \quad(\epsilon\in\{+1,-1\})               \label{eq:abstract-induction}
\end{equation}
for every $k\geq0$.  The case $k=0$ is immediate.  Suppose the claim holds
for $k$, let $j\geq\ell(k+1)$, and fix $\epsilon\in\{+1,-1\}$.  For
$0\leq i\leq\ell$ we have $j-i\geq\ell k$, so
$\shift_{\epsilon kq}S_\epsilon\sqsubset x_{j-i}=H^i(x_j)$; equivalently,
$S_\epsilon\sqsubset H^i(\shift_{-\epsilon kq}x_j)$.  The second line of
\eqref{eq:abstract-pump}, applied to $y=\shift_{-\epsilon kq}x_j$, therefore
gives $\shift_{\epsilon q}S_\epsilon\sqsubset\shift_{-\epsilon kq}x_j$, that
is, $\shift_{\epsilon(k+1)q}S_\epsilon\sqsubset x_j$.  This proves the claim.
Taking $j=n$ and $k=0,\ldots,m$, and recalling that $U\sqsubset x_n$, gives
$W_m\sqsubset y$ whenever $n\geq\ell m$.  This proves~\textup{(a)}.

It remains to pass to $\Om_H$.  Choose a point $x\in\Om_H$ containing $U$.
For each $m$, surjectivity of $H|_{\Om_H}$ supplies $y\in\Om_H$ with
$H^{\ell m}(y)=x$; by~\textup{(a)}, $W_m\sqsubset y$, so $W_m$ is a
consistent finite pattern and $\cyl{W_m}\cap\Om_H$ is nonempty.  The
cylinders $\cyl{W_m}\cap\Om_H$ are nested and compact, so their intersection
contains a point of $\Om_H$ extending $W$.  Thus $W$ is consistent and is
extended by a point of $\Om_H$, which proves~\textup{(b)}.  Finally, let $D\subset
\operatorname{dom}(W)$ be finite; then $D\subset\operatorname{dom}(W_m)$ for
some $m$, and~\textup{(a)}, applied with $x=H^n(y)$, shows that
$(H|_{\Om_H})^{-n}(\cyl U\cap\Om_H)\subset\cyl{W|_D}\cap\Om_H$ for every
$n\geq\ell m$.  Hence $W$ is a uniform ancestral signature of $U$, proving
\textup{(c)}.
\end{proof}

The second part of the ancestral-pumping reduction turns a signature,
together with an incompatibility, into non-transitivity.

\begin{proposition}[Common-ancestor separation]
\label{prop:ancestral-obstruction}
Suppose a finite pattern $U$ occurs in $\Om_H$ and has a uniform ancestral
signature $W$.  If a symmetry $\tau$ of the lattice commutes with $H$ and a
pattern $B\sqsubset W$ is incompatible with $\shift_v\tau B$ for every
$v\in\Z^d$, then the spatial action on $\Om_H$ is not topologically
transitive.  In fact, the orbit saturations of $\cyl U\cap\Om_H$ and
$\cyl{\tau U}\cap\Om_H$ are disjoint nonempty shift-invariant relatively open
subsets of $\Om_H$.
\end{proposition}

\begin{proof}
Every subpattern of a uniform ancestral signature is again a uniform ancestral
signature.  Since $H$ commutes with both translations and $\tau$, applying
$\shift_v\tau$ preserves the uniform ancestral signature property.  Thus
$\shift_v\tau B$ is a signature of $\shift_v\tau U$.  Since $B$
and $\shift_v\tau B$ are incompatible, there is a cell $p$ in their common
domain with $B(p)\neq\shift_v\tau B(p)$.  Let $N$ be large enough that both
signature conditions hold for $D=\{p\}$.  If some $z\in\Om_H$ contained both
$U$ and $\shift_v\tau U$, then surjectivity of $H|_{\Om_H}$ would supply
$y\in\Om_H$ with $H^N(y)=z$, and the two signatures would force $y(p)=B(p)$
and $y(p)=\shift_v\tau B(p)$ simultaneously, a contradiction.  Hence
$\cyl U\cap\Om_H$ and $\shift_v(\cyl{\tau U}\cap\Om_H)$ are disjoint for every
$v\in\Z^d$.  Let $A=\cyl U\cap\Om_H$ and
$C=\cyl{\tau U}\cap\Om_H$.  If their orbit saturations intersected, then
$\shift_aA\cap\shift_bC\neq\varnothing$ for some $a,b\in\Z^d$.
Translating by $-a$ would give $A\cap\shift_{b-a}C\neq\varnothing$,
contradicting the preceding conclusion.  Thus their orbit saturations are
disjoint nonempty shift-invariant relatively open subsets of $\Om_H$;
equivalently, the definition of transitivity fails.
\end{proof}

We now state the Life-specific input to the reduction.

\begin{lemma}[Separating pattern]
\label{lem:separating}
There exists a finite pattern $U$ occurring in $\Om$ such that, with respect
to $h=\Life^2$, $U$ has a uniform ancestral signature $W$ containing a
pattern $B$ that is incompatible with $\shift_v\rot B$ for every
$v\in\Z^2$.
\end{lemma}

\noindent\emph{Proof.} See Section~\ref{sec:separating},
page~\pageref{proof:separating}.

\begin{proof}[Proof of Theorem~\ref{thm:nontransitive}]
By Lemma~\ref{lem:separating}, choose $U$, $W$, and $B$ such that $U$ occurs
in $\Om$, $W$ is a uniform ancestral signature of $U$, and $B\sqsubset W$ is
incompatible with $\shift_v\rot B$ for every $v\in\Z^2$.  These are precisely
the hypotheses of Proposition~\ref{prop:ancestral-obstruction} with
$H=h=\Life^2$ and $\tau=\rot$.  Since $\Om_h=\Om$ and $\rot$ commutes with
$h$, that proposition proves the claim.
\end{proof}

\subsection{The construction}
\label{sec:construction}

Here we construct a finite pattern $U$, a periodic lane $B$, and two
subpatterns $S_+,S_-\sqsubset U$; the properties required by
Lemma~\ref{lem:separating} are established in Section~\ref{sec:separating}.
Appendix~\ref{app:interfaces} collects the exact coordinates, phases, cell
data, and reconstruction rules for all these patterns.
The construction has two requirements.  First, $U$ must occur in the limit
set.  We ensure this by taking it from a period-two configuration $c_0$
that joins a marching-band channel to two Kynn\"os half-planes.  The
channel provides the lane and its pump states; the half-planes provide
room for finite self-forcing Kynn\"os patches.  The transition strips make
these pieces compatible with the same period-two orbit.

Second, every two-step predecessor of a configuration containing $U$ must
recover enough information to start the pump again.  The self-forcing
Kynn\"os patches supply a rail $R$ whose values are inherited by both
predecessor layers.  The marker-forcing calculation uses this inherited
rail together with $U$ in the output layer to recover $U$ in the older
layer.  Since $U$ contains both pump states, this restoration supplies the
seeds throughout a backward history.  The separate pump implication then
extends those seeds by one lane period at a time.  The periodic
configuration witnesses that the chosen patterns are realizable; the
forcing lemmas apply to arbitrary configurations satisfying their finite
premises.  We now specify the pieces used to meet these requirements.

The construction uses two spatiotemporally periodic configurations from
\cite[Section~3]{salo-torma2022oracles}.  The first is the Kynn\"os fixed
point, of spatial period $6\times6$; we fix a finite patch
$K\sqsubset\text{Kynn\"os}$, obtained from a rectangle by removing specified
cells near its corners, whose self-forcing property is recorded as
Fact~\ref{fact:kynnos}.  The second is the marching band introduced in
\cite[Section~3]{salo-torma2022oracles}, a configuration of temporal period two.

We now build $c_0$ by placing a finite-width central channel of the marching
band between two Kynn\"os half-planes, with a transition strip on each side.
The Kynn\"os background has vertical period $6$ and the marching band has
vertical period $4$; the transition strips use their common period $12$, so
the resulting configuration $c_0$ is vertically $12$-periodic.  The
width-three lane used below is a restriction of this channel.
By Lemma~\ref{lem:period}, $c_0$ and $c_1=\Life(c_0)$ form a period-two
orbit, and hence both lie in $\Om$.

The pattern $U\sqsubset c_0$, shown in Figure~\ref{fig:marker}, surrounds the
two ends of a channel segment.  Its lateral blocks lie in the Kynn\"os
half-planes, and its central portion contains the lane and the pump states.
The pattern $U$
contains four translates of $K$; write $R$ for their union and call it the
\emph{rail}.  These four patches supply the inherited constraints used in
Lemma~\ref{lem:marker}.  Using $h=\Life^2$ keeps the construction in the same
phase of the period-two background while the lane grows.

Next we record the lane.  On columns $15,16,17$ the configuration $c_0$ has
the period-four row words
\begin{equation}
  L_0=000,\qquad L_1=001,\qquad
  L_2=001,\qquad L_3=101,                    \label{eq:lane-words}
\end{equation}
and we let $B$ denote the resulting bi-infinite vertical lane,
\begin{equation}
  B(15+i,y)=L_{y\bmod4}(i),
  \qquad i\in\{0,1,2\},\ y\in\Z.             \label{eq:lane-definition}
\end{equation}
The marker $U$ contains $B$ on rows $-14$ through $33$.

Finally, we define the pump states.  Let $q=(0,4)$, let $S_+$ be the
restriction of $c_0$ to $[9,22]\times[28,31]$, and let $\psi$ be reflection
in the $x$-axis followed by translation by $(4,19)$, i.e.,
\[
  \psi=\shift_{(4,19)}\refl,
  \qquad S_-=\psi S_+.
\]
Then $S_-$ is supported on $[13,26]\times[-12,-9]$.  A direct comparison of
the specified positions and cell values gives
\begin{equation}
  S_+\sqsubset U,
  \qquad S_-\sqsubset U.                     \label{eq:states-in-marker}
\end{equation}
The two rectangles are shown in Figure~\ref{fig:seeds}.
The linear part of $\psi$
sends $q$ to $-q$, so $\psi\shift_q=\shift_{-q}\psi$.

\begin{figure}[ht]
\centering
\input{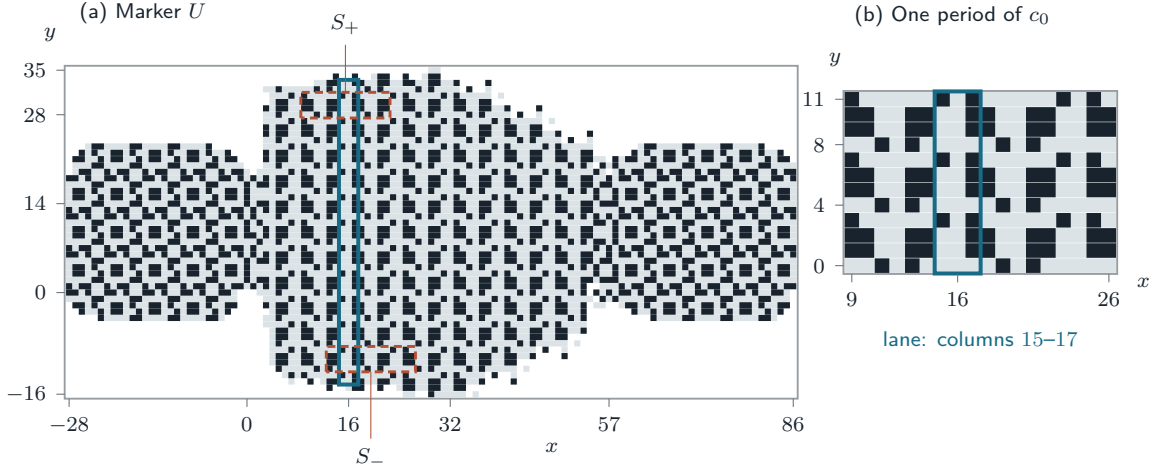}
\caption{The pattern $U$, with the pump states $S_+$ and $S_-$ outlined by
dashed rectangles, together with one vertical period of $c_0$ on columns
$9$--$26$, enlarged.  Blue outlines the lane on columns $15$--$17$.
Black cells are live, pale cells are fixed dead, and white regions are
unspecified.  Coordinates refer to cell centres.}
\label{fig:marker}
\end{figure}

\begin{figure}[ht]
\centering
\input{generated/lane-seeds.tex}
\caption{The two pump states contained in $U$.  Their
blue subpatterns are periods of the same lane, and $S_-$ is obtained by
reflecting $S_+$ in the $x$-axis and then translating it.  Pale cells are
fixed dead and black cells are live.}
\label{fig:seeds}
\end{figure}
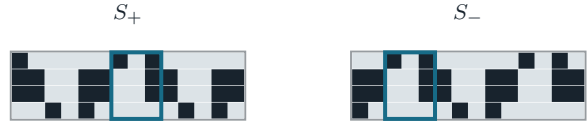

\FloatBarrier

\subsection{The separating lemma}
\label{sec:separating}

We use the published Kynn\"os self-forcing property, three finite lemmas
established in Appendix~\ref{app:certificates}, and a direct crossing
argument.  Table~\ref{tab:interfaces} identifies the formal statement and
verification method for each finite ingredient.  Figure~\ref{fig:separating-dag}
then shows how these statements, together with
Theorem~\ref{thm:pumping}, prove Lemma~\ref{lem:separating}.  Statements whose
detailed finite proofs are deferred to Appendix~\ref{app:certificates} carry
explicit page links below, so they may be read as black boxes on a first
reading.

\begin{table}[ht]
\centering
\small
\begin{tabularx}{\textwidth}{@{}>{\raggedright\arraybackslash}p{.19\textwidth}
  >{\raggedright\arraybackslash}p{.15\textwidth}X
  >{\raggedright\arraybackslash}p{.18\textwidth}@{}}
\toprule
Proof role & Formal result & Finite statement & Verification\\
\midrule
occurrence of $U$ in $\Om$
  & Lemma~\ref{lem:period}
  & $c_0$ and $c_1$ form a period-two orbit containing $U$
  & direct periodicity check\\
marker forcing, used for persistence
  & Lemma~\ref{lem:marker}
  & $U$ and the inherited rail force $U$ two steps earlier
  & finite-state and SAT\\
ancestral pumping
  & Lemma~\ref{lem:pump}
  & three occurrences of a pump state force its outward translate
  & finite-state and SAT\\
crossing incompatibility
  & Lemma~\ref{lem:crossing}
  & perpendicular period-four lanes disagree on every crossing
  & direct $3\times3$ argument\\
\bottomrule
\end{tabularx}
\caption{The finite ingredients used to prove Lemma~\ref{lem:separating}.
The two computational routes use the same input patterns but different
encodings and checkers.}
\label{tab:interfaces}
\end{table}

\FloatBarrier

\begin{figure}[ht]
\centering
\begin{tikzpicture}[font=\scriptsize\sffamily,
  every node/.style={proofbox,text width=2.65cm,minimum height=0.9cm}]
  \node (fact) at (-4.5,1.5)
    {Fact~\ref{fact:kynnos}\\Kynn\"os self-forcing};
  \node (marker) at (-4.5,0.1)
    {Lemma~\ref{lem:marker}\\marker forcing};
  \node (period) at (-0.9,2.2)
    {Lemma~\ref{lem:period}\\period-two witness};
  \node (persist) at (-0.9,0.8)
    {Lemma~\ref{lem:persistence}\\persistence};
  \node (pumpfinite) at (-0.9,-0.6)
    {Lemma~\ref{lem:pump}\\lane pump};
  \node (pumptheorem) at (-0.9,-2.0)
    {Theorem~\ref{thm:pumping}\\ancestral pumping};
  \node (cross) at (-0.9,-3.4)
    {Lemma~\ref{lem:crossing}\\crossing};
  \node (separating) at (3.6,-0.6)
    {Lemma~\ref{lem:separating}\\separating pattern};
  \draw[flowarrow] (fact.east) -- (persist.west);
  \draw[flowarrow] (marker.east) -- (persist.west);
  \draw[flowarrow] (persist.east) -- (separating.west);
  \draw[flowarrow] (period.east) -- (separating.west);
  \draw[flowarrow] (pumpfinite.east) -- (separating.west);
  \draw[flowarrow] (pumptheorem.east) -- (separating.west);
  \draw[flowarrow] (cross.east) -- (separating.west);
\end{tikzpicture}
\caption{Proof dependencies used in Section~\ref{sec:separating}.  Every
arrow $A\to B$ means that statement $A$ is used in the proof of statement
$B$.}
\label{fig:separating-dag}
\end{figure}
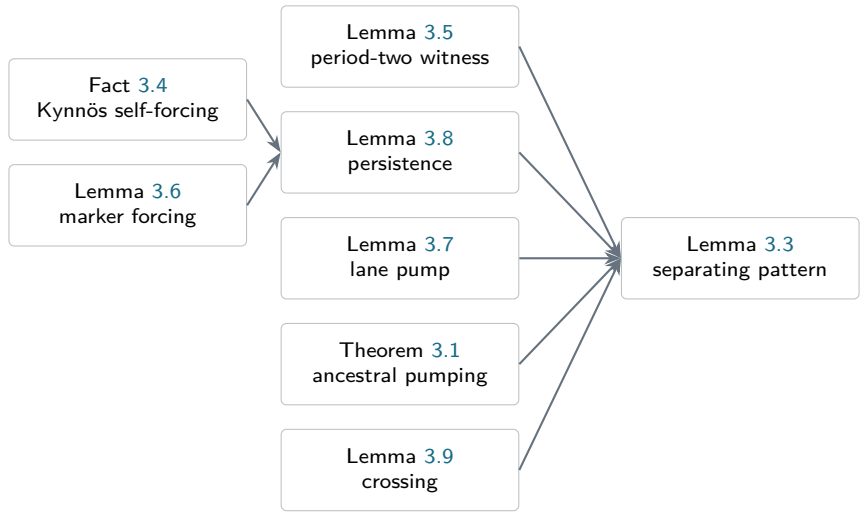

\begin{fact}[Kynn\"os self-forcing]\label{fact:kynnos}
For every $v\in\Z^2$ and $x,y\in X$,
\[
 \Life(y)=x,\quad \shift_vK\sqsubset x
 \quad\Longrightarrow\quad \shift_vK\sqsubset y.
\]
Thus, whenever a configuration contains a translate $\shift_vK$, every
one-step predecessor contains the same translate.
\end{fact}

\noindent\emph{Proof.} This is the Kynn\"os self-forcing result of Salo and
T\"orm\"a~\cite[Lemma~21]{salo-torma2022oracles}.

\begin{lemma}[Period-two witness]\label{lem:period}
The configuration $c_0$ is vertically $12$-periodic and satisfies
$\Life(c_0)=c_1$ and $\Life(c_1)=c_0$.
\end{lemma}

\noindent\emph{Proof.} See Appendix~\ref{app:interfaces},
page~\pageref{proof:period}.

\begin{lemma}[Marker forcing]\label{lem:marker}
Let $x,w,y\in X$ with $\Life(w)=x$ and $\Life(y)=w$.  Then
\begin{equation}
 U\sqsubset x,\quad R\sqsubset w,\quad R\sqsubset y
 \quad\Longrightarrow\quad U\sqsubset y.       \label{eq:marker-force}
\end{equation}
\end{lemma}

\noindent\emph{Proof.} See Appendix~\ref{app:finite-state},
page~\pageref{proof:marker}.  Appendix~\ref{app:clausal},
page~\pageref{proof:sat}, gives an independent SAT verification.

\begin{lemma}[Lane pump]\label{lem:pump}
Let $x_0,x_1,x_2\in X$ with $h(x_1)=x_0$ and $h(x_2)=x_1$.  Then
\begin{equation}
 S_+\sqsubset x_0,\ S_+\sqsubset x_1,\ S_+\sqsubset x_2
 \quad\Longrightarrow\quad
 \shift_qS_+\sqsubset x_2.                             \label{eq:single-pump}
\end{equation}
\end{lemma}

\noindent\emph{Proof.} See Appendix~\ref{app:finite-state},
page~\pageref{proof:pump}.  Appendix~\ref{app:clausal},
page~\pageref{proof:sat}, gives an independent SAT verification.

\begin{lemma}[Persistence]\label{lem:persistence}
If $x,y\in X$, $h(y)=x$, and $U\sqsubset x$, then $U\sqsubset y$.
\end{lemma}

\begin{proof}
Put $w=\Life(y)$, so $\Life(w)=x$.  Apply Fact~\ref{fact:kynnos} at each of
the four translates of $K$ in $U$, first from $x$ to $w$ and then from $w$
to $y$, keeping each translation vector fixed.  This shows that $R\sqsubset w$ and $R\sqsubset y$.  Now
Lemma~\ref{lem:marker} gives $U\sqsubset y$.
\end{proof}

It remains to record that the lane cannot coexist with its quarter-turn.  A
translated quarter-turn $\shift_{(a,b)}\rot B$ assigns at $(x,y)$ the value
that $B$ assigns at $(b-y,x-a)$; its horizontal support therefore meets the
vertical support of $B$ in the square
\begin{equation}
  \{15,16,17\}\times\{b-17,b-16,b-15\}.       \label{eq:crossing-square}
\end{equation}

\begin{lemma}[Crossing]\label{lem:crossing}
For every $(a,b)\in\Z^2$, the two lane patterns $B$ and
$\shift_{(a,b)}\rot B$ disagree at a cell of the square
\eqref{eq:crossing-square}.
\end{lemma}

\begin{proof}
Choose $r,s\in\{0,1,2,3\}$ with
\[
  r\equiv b-17\pmod4,
  \qquad s\equiv15-a\pmod4.
\]
At
$(x,y)=(15+i,b-17+j)$, for $0\leq i,j<3$, the vertical and horizontal
lanes induce two $3\times3$ blocks $C^{\mathrm v}$ and $C^{\mathrm h}$ on
the crossing square, with row index $j$ and column index $i$.  Their entries
are
\begin{equation}
 C^{\mathrm v}_{j,i}=L_{r+j}(i),
 \qquad C^{\mathrm h}_{j,i}=L_{s+i}(2-j),                   \label{eq:crossing-values}
\end{equation}
with subscripts modulo four.

Suppose $C^{\mathrm v}=C^{\mathrm h}$.  Every $L_k(1)$ is zero, so the
middle row of $C^{\mathrm h}$ is zero.  Equality therefore gives
$L_{r+1}=000$.  Since $L_0$ is the only all-zero word among
$L_0,\ldots,L_3$, we have $r\equiv3\pmod4$.  Likewise, the middle column of
$C^{\mathrm v}$ is zero, while the middle column of $C^{\mathrm h}$ is
$L_{s+1}$ in reverse order.  Thus $L_{s+1}=000$ and
$s\equiv3\pmod4$.  But at the cell with $(i,j)=(0,2)$,
\eqref{eq:crossing-values} gives $C^{\mathrm v}_{2,0}=L_1(0)=0$ and $C^{\mathrm h}_{2,0}=L_3(0)=1$, a
contradiction.
\end{proof}

\begin{proof}[Proof of Lemma~\ref{lem:separating}]
\phantomsection\label{proof:separating}
Lemma~\ref{lem:period}, together with $U\sqsubset c_0$, shows that $U$ occurs
in $\Om$.  Lemma~\ref{lem:persistence} gives the persistence hypothesis
$U\sqsubset h(y)\Rightarrow U\sqsubset y$ of
Theorem~\ref{thm:pumping}.  Lemma~\ref{lem:pump} gives its pumping hypothesis
for $S_+$ with $H=h$, $\ell=2$, and $q=(0,4)$.  Conjugating that implication
by $\psi$ gives the corresponding hypothesis for $S_-$ with displacement
$-q$.  Theorem~\ref{thm:pumping} therefore supplies a uniform ancestral
signature $W$ of $U$.

The pump states $S_+$ and $S_-$ provide complete lane periods on rows
$28,\ldots,31$ and $-12,\ldots,-9$, respectively, while $U$ contains the
intervening segment.  Repeated pumping therefore fills the entire
bi-infinite lane, so $B\sqsubset W$.  Finally,
Lemma~\ref{lem:crossing} says that $B$ is incompatible with
$\shift_v\rot B$ for every $v\in\Z^2$.  Thus $U$, $W$, and $B$ have all the
properties asserted in Lemma~\ref{lem:separating}.
\end{proof}

The same construction quantifies the separation at finite depth.  Let
\[
  U_{\mathrm v}=\shift_{(-16,-9)}U,
  \qquad
  U_{\mathrm h}=\rot U_{\mathrm v},
\]
so that the lane of $U_{\mathrm v}$ is vertical on columns $\{-1,0,1\}$ and
that of $U_{\mathrm h}$ is horizontal on rows $\{-1,0,1\}$.

\begin{corollary}[Quantitative separation]
\label{cor:quantitative}
For every $m\geq1$ and $v=(a,b)\in\Z^2$ with $\|v\|_\infty\leq4m+20$,
\[
  \cyl{U_{\mathrm v}}\cap\cyl{\shift_vU_{\mathrm h}}
  \cap\Life^{4m}(X)=\varnothing.
\]
\end{corollary}

\begin{proof}
Let $x$ lie in the intersection, and choose $y\in X$ with $\Life^{4m}(y)=x$,
so that $h^{2m}(y)=x$.  Applying Theorem~\ref{thm:pumping} with $H=h$ and
$\ell=2$ to $U_{\mathrm v}\sqsubset x$ shows that $y$ contains
the vertical lane on
\[
  \{-1,0,1\}\times[-21-4m,\,22+4m].
\]
Indeed, after translating $U$ to $U_{\mathrm v}$, the lower and upper pump
states occupy rows $[-21,-18]$ and $[19,22]$, respectively; each pump step
extends the lane by four rows at each end.
Since $\rot$ and translations commute with $h$, the rotated and translated
marker satisfies the same persistence and pump hypotheses.  Applying the
theorem to the occurrence of $\shift_vU_{\mathrm h}$ in $x$ shows that
$y$ also contains the horizontal lane on
\[
  [a-21-4m,\,a+22+4m]\times\{b-1,b,b+1\}.
\]
If $\|v\|_\infty\leq4m+20$, then $a,b\in[-20-4m,20+4m]$, so the second
segment contains the columns $\{-1,0,1\}$ and the first contains the rows
$\{b-1,b,b+1\}$.  The two segments therefore contain the square
$\{-1,0,1\}\times\{b-1,b,b+1\}$.  Since $U_{\mathrm h}=\rot U_{\mathrm v}$, any translate of the lane of
$U_{\mathrm v}$ and any
translate of its quarter-turn are incompatible by Lemma~\ref{lem:crossing};
so the two lane segments disagree somewhere in that square.  Both are
restrictions of the single configuration $y$, a contradiction.
\end{proof}

\section{Discussion}
\label{sec:discussion}

Corollary~\ref{cor:quantitative} makes the finite-to-infinite transition
quantitative: increasing the available predecessor depth by four Life
generations adds one four-row period to each end of the forced lane.  Passing
to configurations with predecessors of every depth turns these expanding
finite exclusions into exclusion at every displacement.  The certificates
establish fixed local implications, iteration supplies the unbounded spatial
scale, and compactness realizes the resulting infinite signature in $\Om$.

The conclusion is a failure of transitivity, not a complete decomposition of
$\Om$.  The two disjoint orbit-saturated open sets constructed in
Proposition~\ref{prop:ancestral-obstruction} do not partition the limit set;
for example, the all-zero configuration belongs to neither.  Nor does the
ancestral-pumping criterion classify the cellular automata to which it
applies.  Its present application relies on three Life-specific ingredients:
a realizable periodic background, a finite marker restored throughout a
backward history, and pump states whose forced signatures are incompatible
after rotation.

Several questions therefore remain open.  The marker is explicit but large,
and we do not know the smallest possible separating pair.  It would also be
interesting to determine which cellular automata admit persistent periodic
backgrounds and pumps of this kind, and whether other mechanisms can produce
fixed separating pairs without periodic lanes.

\paragraph{Author contributions.}
All authors contributed equally to this work and are listed in alphabetical
order.  Ziran Li is the corresponding author.

\paragraph{Funding.}
This research received no external funding.

\paragraph{Competing interests.}
The authors declare none.

\paragraph{Data availability.}
The finite input patterns and proof certificates are archived at
\url{https://doi.org/10.5281/zenodo.23003844}.
The accompanying verification package contains the reconstruction scripts,
finite-state calculations, certificate checkers, and a file manifest.
Appendix~\ref{app:certificates} gives the verification command.

\section*{Acknowledgments}

\paragraph{AI assistance.}
The authors used OpenAI's GPT models, primarily GPT-5.6, to assist with
mathematical exploration, proof development, software implementation, and
manuscript preparation, including drafting and revision. The authors
reviewed and revised the mathematical arguments and the text. Computational
checks were performed by running the verification scripts and certificate
checkers described in Appendix~\ref{app:certificates}. The authors take
responsibility for the arguments, computations, and final text.

\appendix
\small
\input{certificate_appendix.tex}

{\scriptsize
\bibliographystyle{unsrt}
\bibliography{references}
}

\end{document}

%% file: generated/lane-seeds.tex
\begin{tikzpicture}
  \begin{scope}[x=0.22cm,y=0.22cm,xshift=0cm]
    \path[fill=fixedcell,draw=none]
  (-0.48,-0.48) rectangle (13.48,0.48)
  (-0.48,0.52) rectangle (13.48,1.48)
  (-0.48,1.52) rectangle (13.48,2.48)
  (-0.48,2.52) rectangle (13.48,3.48);
    \path[fill=livecell,draw=none]
  (1.52,-0.48) rectangle (2.48,0.48)
  (3.52,-0.48) rectangle (4.48,0.48)
  (9.52,-0.48) rectangle (10.48,0.48)
  (11.52,-0.48) rectangle (12.48,0.48)
  (-0.48,0.52) rectangle (1.48,1.48)
  (3.52,0.52) rectangle (5.48,1.48)
  (7.52,0.52) rectangle (9.48,1.48)
  (11.52,0.52) rectangle (13.48,1.48)
  (-0.48,1.52) rectangle (1.48,2.48)
  (3.52,1.52) rectangle (5.48,2.48)
  (7.52,1.52) rectangle (9.48,2.48)
  (11.52,1.52) rectangle (13.48,2.48)
  (-0.48,2.52) rectangle (0.48,3.48)
  (5.52,2.52) rectangle (6.48,3.48)
  (7.52,2.52) rectangle (8.48,3.48);
    \draw[markerframe] (-0.55,-0.55) rectangle (13.55,3.55);
    \draw[laneframe] (5.50,-0.50) rectangle
      (8.50,3.50);
    \node[figurelabel,anchor=south] at (6.50,4.6)
      {$S_+$};
  \end{scope}
  \begin{scope}[x=0.22cm,y=0.22cm,xshift=4.5cm]
    \path[fill=fixedcell,draw=none]
  (-0.48,-0.48) rectangle (13.48,0.48)
  (-0.48,0.52) rectangle (13.48,1.48)
  (-0.48,1.52) rectangle (13.48,2.48)
  (-0.48,2.52) rectangle (13.48,3.48);
    \path[fill=livecell,draw=none]
  (-0.48,-0.48) rectangle (0.48,0.48)
  (5.52,-0.48) rectangle (6.48,0.48)
  (7.52,-0.48) rectangle (8.48,0.48)
  (-0.48,0.52) rectangle (1.48,1.48)
  (3.52,0.52) rectangle (5.48,1.48)
  (7.52,0.52) rectangle (9.48,1.48)
  (11.52,0.52) rectangle (13.48,1.48)
  (-0.48,1.52) rectangle (1.48,2.48)
  (3.52,1.52) rectangle (5.48,2.48)
  (7.52,1.52) rectangle (9.48,2.48)
  (11.52,1.52) rectangle (13.48,2.48)
  (1.52,2.52) rectangle (2.48,3.48)
  (3.52,2.52) rectangle (4.48,3.48)
  (9.52,2.52) rectangle (10.48,3.48)
  (11.52,2.52) rectangle (12.48,3.48);
    \draw[markerframe] (-0.55,-0.55) rectangle (13.55,3.55);
    \draw[laneframe] (1.50,-0.50) rectangle
      (4.50,3.50);
    \node[figurelabel,anchor=south] at (6.50,4.6)
      {$S_-$};
  \end{scope}
\end{tikzpicture}

%% file: certificate_appendix.tex
\section{Verification of the finite statements}
\label{app:certificates}

This appendix records the finite data and coordinates used in the
verification (Section~\ref{app:interfaces}), the finite-state verification of
marker forcing and lane pumping (Section~\ref{app:finite-state}), and the
SAT verification of those implications together with the Kynn\"os
self-forcing implication (Section~\ref{app:clausal}).  The finite-state and
SAT calculations are two separate verification routes based on the same
finite input patterns.  Their independence is therefore at the level of
encodings and checkers, not at the level of input data.  The verification
reported here depends on the explicit finite implications,
input data, and certificates reconstructed and checked below.

After the setup described in the verification package README, run \verb|python3 verify.py --all| from the artifact
root.  This command reconstructs the formulas, checks the certificates, and
recomputes the finite-state proofs.  Data filenames below are relative to
\path{verification/results/}.  This is computational verification, not a
proof-assistant formalization of the whole argument.

\subsection{Finite patterns}
\label{app:interfaces}

\paragraph{The periodic witness.}
The following row words specify the Kynn\"os tile $Q$, the marching-band
tile $M$, and the two transition strips $A_{\mathrm L},A_{\mathrm R}$ of
$c_0$.  Rows and columns are indexed from zero; digits in each word run
from left to right.
The fundamental vertical periods are $6$ for $Q$, $4$ for $M$, and $12$ for
$A_{\mathrm L},A_{\mathrm R}$; blank entries indicate that no additional row
word is needed beyond the corresponding fundamental period.
\begin{center}
\begin{tabular}{@{}rcccc@{}}
\toprule
row & $Q$ & $M$ & $A_{\mathrm L}$ & $A_{\mathrm R}$\\
\midrule
0 & \texttt{001101} & \texttt{10000010} & \texttt{010001} & \texttt{101010}\\
1 & \texttt{001011} & \texttt{11001100} & \texttt{010101} & \texttt{110000}\\
2 & \texttt{110000} & \texttt{11001100} & \texttt{101101} & \texttt{110111}\\
3 & \texttt{010110} & \texttt{00101000} & \texttt{100000} & \texttt{001001}\\
4 & \texttt{100110} & & \texttt{100001} & \texttt{101001}\\
5 & \texttt{110000} & & \texttt{101101} & \texttt{110110}\\
6 & & & \texttt{010101} & \texttt{110010}\\
7 & & & \texttt{010000} & \texttt{001010}\\
8 & & & \texttt{001001} & \texttt{101101}\\
9 & & & \texttt{110101} & \texttt{110001}\\
10 & & & \texttt{100101} & \texttt{110001}\\
11 & & & \texttt{101000} & \texttt{001101}\\
\bottomrule
\end{tabular}
\end{center}
With the first subscript denoting the row, define
\[
c_0(x,y)=
\begin{cases}
Q_{y\bmod6,\,x\bmod6}, & x<0,\\
(A_{\mathrm L})_{y\bmod12,\,x}, & 0\leq x<6,\\
M_{y\bmod4,\,(x-5)\bmod8}, & 6\leq x<53,\\
(A_{\mathrm R})_{y\bmod12,\,x-53}, & 53\leq x<59,\\
Q_{y\bmod6,\,(x-57)\bmod6}, & x\geq59.
\end{cases}
\]
\begin{proof}[Proof of Lemma~\ref{lem:period}]
\phantomsection\label{proof:period}
Set $c_1=\Life(c_0)$.  Direct evaluation gives
$\Life^2(c_0)(x,y)=c_0(x,y)$ at all $900$ cells of
$[-8,66]\times[0,11]$, and $c_1\neq c_0$.
Twelve rows suffice because $c_0$ is vertically $12$-periodic.  Since
$\Life^2$ has radius two, the transition and splice region is fully covered
by the central checked columns $-2,\ldots,60$.  Outside this region the two
Kynn\"os tails are horizontally $6$-periodic, and the six representative
columns $-8,\ldots,-3$ and $61,\ldots,66$ cover those tails together with
their radius-two dependencies.  Thus $\Life(c_1)=c_0$ everywhere.  The
inequality $c_1\neq c_0$ shows that the orbit has exact period two.
\end{proof}

\paragraph{The Kynn\"os patch.}
Let $E_K$ consist of the $32$ cells listed below, and set
\[
 D_K=([2,29]\times[2,23])\setminus E_K,
 \qquad K(x,y)=Q_{y\bmod6,\,x\bmod6}\quad((x,y)\in D_K).
\]
\begin{center}
\begin{tabular}{@{}rl@{\qquad}rl@{}}
\toprule
$y$ & deleted $x$ & $y$ & deleted $x$\\
\midrule
2 & $2,3,4,5,6,27,28,29$ & 19 & $2$\\
3 & $2,3,28,29$ & 20 & $2$\\
4 & $2,29$ & 21 & $2,29$\\
5 & $29$ & 22 & $2,3,28,29$\\
6 & $29$ & 23 & $2,3,4,25,26,27,28,29$\\
\bottomrule
\end{tabular}
\end{center}
This $584$-cell pattern is the spatially phased Kynn\"os patch to which
\cite[Lemma~21]{salo-torma2022oracles} applies.  The SAT proof below checks
its self-forcing property at the same position in the predecessor and output.

\paragraph{The marker and rail.}
Let $U_{\mathrm{raw}}$ be the pattern in the field
\texttt{surviving\_shared\_pattern} of
\path{marker-source.json}.  Define
\[
\begin{split}
 E_7={}&\{(54,-2),(57,-2),(58,-2),(54,1),(55,1),(54,0),(55,-1)\},\\
\widehat U={}&U_{\mathrm{raw}}|_{\operatorname{dom}(U_{\mathrm{raw}})\setminus E_7},\\
U_{\mathrm b}={}&\widehat U|_{y\leq13},\qquad
U_{\mathrm t}=\shift_{(0,-24)}(\widehat U|_{y\geq38}),\qquad
U=U_{\mathrm b}\cup U_{\mathrm t}.
\end{split}
\]
We call $U_{\mathrm b}$ and $U_{\mathrm t}$ the bottom and top caps,
respectively.  The displayed reconstruction defines $U$.  The named JSON
field is the machine-readable representation used by the verification scripts, and the
artifact manifest identifies its source file by the SHA-256 digest
\begin{center}
\footnotesize\ttfamily
72dc522f1516ef73682a4c7d28876ff1\\
a0f24f688754899183f315312933eb89
\end{center}
and the top-level verification command reconstructs the pattern before
checking the marker and pump implications.
Of the seven omitted cells, four distinguish the predecessor branches used
in the construction, while three remain free under the local constraints.
After these seven cells are removed, the remaining prescribed cells are
common to all predecessor branches retained in the construction.
The raw pattern has $6{,}726$ cells and $\widehat U$ has $6{,}719$.
The two cap domains are disjoint, with $2{,}456$ and $1{,}503$ cells.
Their union has $3{,}959$ cells and bounding box
$[-28,86]\times[-16,35]$.  All its values agree with $c_0$;
translation by $(0,-24)$ preserves the vertical phase.

The rail is
\[
 T=\{(-30,-6),(-30,0),(57,-6),(57,0)\},\qquad
 R=\bigcup_{v\in T}\shift_vK.
\]
The four bounding boxes are
\[
\begin{array}{ll}
[-28,-1]\times[-4,17], & [-28,-1]\times[2,23],\\[2pt]
[59,86]\times[-4,17], & [59,86]\times[2,23].
\end{array}
\]
The translates agree on overlaps, and their union is a $1{,}504$-cell
subpattern of $U$.

\paragraph{The pump states and coordinate conversion.}
The states $S_+$ and $S_-$ are the restrictions of $c_0$ to
$[9,22]\times[28,31]$ and $[13,26]\times[-12,-9]$, respectively.
Each has $56$ cells, is a subpattern of $U$, and contains four complete rows
of the width-three lane, with $S_-=\shift_{(4,19)}\refl S_+$.
The certificate coordinates are related to the paper coordinates by
\[
 (x,y)_{\mathrm{certificate}}\longmapsto(x,y-28)_{\mathrm{paper}}.
\]
Their seed rows $56,\ldots,59$ therefore become $28,\ldots,31$.
Under this coordinate conversion, the verified upper-pump implication is
exactly \eqref{eq:single-pump}; reflection then gives the lower displacement
$(0,-4)$.

\subsection{Finite-state verification}
\label{app:finite-state}

\paragraph{The method and time convention.}
The finite-state calculation decides whether prescribed finite regions in
successive Life layers can be extended consistently with all local Life
equations while a vertical cut is swept across the spacetime patch.
For the generic transfer write
$V_2\xrightarrow{\Life}V_1\xrightarrow{\Life}V_0$;
each arrow is one Life generation.  Place a vertical cut between columns
$a,a+1$.  The boundary symbol at row $y$ is
\[
 \bigl(V_1(a,y),V_1(a+1,y),V_2(a,y),V_2(a+1,y)\bigr).
\]
Two adjacent columns from each unknown predecessor layer are carried across
the cut because this is the information needed to evaluate the local Life
equations when the next column is exposed.
The symbols are read from bottom to top, and the cut moves from right to
left.  Figure~\ref{fig:boundary-transfer} shows one move.
The output $V_0$ is prescribed on the checked domain; the boundary contains
the two unknown predecessor layers.

A \emph{layered deterministic finite automaton} (DFA) reads one boundary
symbol per row and has one state layer per row.  Its accepted words are the
boundaries that admit extensions into the scanned region.
A \emph{monitor} tracks a property of cells as they pass the cut,
such as a flag for a target mismatch.  A single boundary word may admit
several interior extensions with different monitor outcomes; a terminal
state's \emph{colour} records the set of all such outcomes.

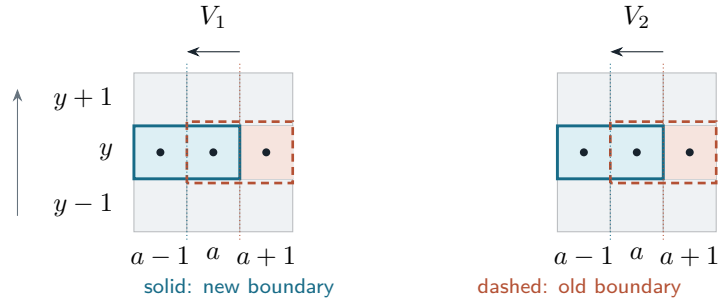
\begin{figure}[ht]
\centering
\begin{tikzpicture}[x=7mm,y=7mm,font=\small]
  \foreach \offset/\layer in {0/1,8/2} {
    \begin{scope}[xshift=\offset*7mm]
      \node at (1.5,4.05) {$V_{\layer}$};
      \foreach \i in {0,1,2} {
        \foreach \j in {0,1,2} {
          \draw[fill=fixedcell!45,draw=ink!30]
            (\i,\j) rectangle +(1,1);
        }
      }
      \fill[lanepale] (0,1) rectangle (2,2);
      \fill[cappale] (2,1) rectangle (3,2);
      \draw[laneblue,line width=1.1pt] (0,1) rectangle (2,2);
      \draw[caprust,densely dashed,line width=1pt]
        (1,0.92) rectangle (3,2.08);
      \foreach \i in {0,1,2} {
        \fill[ink] (\i+0.5,1.5) circle (1.5pt);
      }
      \draw[laneblue,densely dotted] (1,-0.15) -- (1,3.2);
      \draw[caprust,densely dotted] (2,-0.15) -- (2,3.2);
      \draw[-{Stealth[length=1.7mm]},ink] (2,3.4) -- (1,3.4);
      \node[below] at (0.5,-0.1) {$a-1$};
      \node[below] at (1.5,-0.1) {$a$};
      \node[below] at (2.5,-0.1) {$a+1$};
    \end{scope}
  }
  \node[left] at (-0.2,0.5) {$y-1$};
  \node[left] at (-0.2,1.5) {$y$};
  \node[left] at (-0.2,2.5) {$y+1$};
  \draw[-{Stealth[length=1.7mm]},muted] (-2.2,0.3) -- (-2.2,2.7);
  \node[font=\scriptsize\sffamily,text=laneblue,anchor=west]
    at (0,-1.05) {solid: new boundary};
  \node[font=\scriptsize\sffamily,text=caprust,anchor=west]
    at (6.3,-1.05) {dashed: old boundary};
\end{tikzpicture}
\caption{Moving the cut one column left.  At each row, the new boundary
has four bits in columns $a-1,a$ across $V_1,V_2$; the two bits in column
$a+1$ become hidden choices.  Symbols are read upwards.  Three successive
rows determine the Life equations at their middle row in column $a$,
with the prescribed output layer $V_0$ omitted from the drawing.}
\label{fig:boundary-transfer}
\end{figure}

\begin{proposition}[Coloured boundary transfer]\label{prop:colored-transfer}
Suppose each checked Life equation has its centre and all eight neighbours
in the corresponding input domain.  Then a layered DFA can represent
exactly the boundary words that extend into the scanned part of the finite
spacetime patch.  If a deterministic finite monitor reads cells as their
columns pass the cut, terminal colours give the set of monitor outcomes
of those extensions.
\end{proposition}

The transfer is therefore both sound and complete: it neither discards a
valid finite spacetime extension nor introduces a spurious one.

\begin{proof}
Initially only the fixed cells in the two rightmost boundary columns are
checked; no Life equation has passed the cut.  Maintain the following
invariant: an accepted boundary word has exactly the extensions permitted
by the scanned constraints, and its colour is the set of their monitor
outcomes.

To move the cut left, read the four new bits in columns $a-1,a$, and choose
the two hidden bits in column $a+1$.  These hidden choices introduce
nondeterminism, which is handled by the usual subset determinization.  During
this row scan a determinized state is therefore a set of correlated tuples
$(s,r_{-},r)$, where $s$ is a state of the DFA at the previous cut position
and $r_{-},r$ are the preceding two six-bit rows.  At the first rows, a
distinguished boundary marker is used when fewer than two preceding rows are
available.  Monitor information is
stored together with its boundary tuple in the product state, or encoded
in the coloured future language of the old state.

For each tuple and hidden choice, check all fixed values, advance the old
DFA on its induced boundary symbol, and, once three rows are available,
check the two Life equations at the middle row of column $a$ wherever those
equations belong to the patch.  Reject a branch that fails a check;
otherwise advance its monitor and keep the updated tuple.  The extreme rows
supply the complete radius-one input neighbourhood, consisting of the centre
and its eight Moore neighbours (the input halo), for the equations on the
interior rows.

Every extension supplies an accepted sequence of choices.  Conversely,
each accepted sequence satisfies the fixed values and every equation as it
passes the cut, and hence extends the new boundary.  Taking sets of
correlated tuples preserves both statements,
and unioning monitor outcomes at termination gives exactly the asserted
colour.  Merging states with identical coloured future languages preserves
these extensions and colours.  Induction over rows and columns proves the
invariant and the proposition.
\end{proof}

One-step transfers use the corresponding one-layer boundary construction.

\phantomsection\label{proof:pump}
\subsubsection*{Finite-state proof of Lemma~\ref{lem:pump}}
Write the four-generation history as
\[
 Z_4\xrightarrow{\Life}Z_3\xrightarrow{\Life}Z_2
 \xrightarrow{\Life}Z_1\xrightarrow{\Life}Z_0,
 \qquad x_j=Z_{2j}\quad(j=0,1,2).
\]
The middle layer in \eqref{eq:single-pump} is $Z_2$.  The first two-step
transfer uses
$(V_2,V_1,V_0)=(Z_2,Z_1,Z_0)$; the second uses $(Z_4,Z_3,Z_2)$.
The first transfer reduces the possible intermediate behaviour to three
candidate patterns on columns $10,\ldots,21$: the two nonreference candidates
$T_{\mathrm L},T_{\mathrm R}$ and the reference candidate
$T_{\mathrm{ref}}\sqsubset c_0$.  A dot means an unspecified cell:
\begin{center}
\begin{tabular}{@{}ccc@{}}
\toprule
pattern & row $32$ & row $33$\\
\midrule
$T_{\mathrm L}$ & \texttt{..0100......} & \texttt{..101.......}\\
$T_{\mathrm R}$ & \texttt{...1000.111.} & \texttt{............}\\
$T_{\mathrm{ref}}$ & \texttt{0101000...01} & \texttt{.0011...100.}\\
\bottomrule
\end{tabular}
\end{center}
Under the displayed Life equations, the boundary scans establish three
implications:
\begin{align*}
 S_+\sqsubset Z_0,\qquad S_+\sqsubset Z_2
 &\ \Longrightarrow\
 T_{\mathrm L}\sqsubset Z_2\ \text{or}\
 T_{\mathrm R}\sqsubset Z_2\ \text{or}\
 T_{\mathrm{ref}}\sqsubset Z_2,\\
 S_+\sqsubset Z_2,\qquad S_+\sqsubset Z_4
 &\ \Longrightarrow\
 T_{\mathrm L}\not\sqsubset Z_2\ \text{and}\
 T_{\mathrm R}\not\sqsubset Z_2,\\
 S_+\cup T_{\mathrm{ref}}\sqsubset Z_2,\quad S_+\sqsubset Z_4
 &\ \Longrightarrow\ \shift_qS_+\sqsubset Z_4.
\end{align*}
For the first implication, the monitor assigns failure bits $1,2,4$ to
$T_{\mathrm L},T_{\mathrm R},T_{\mathrm{ref}}$ respectively.  Its reachable
terminal masks are $\{3,4,5,6\}$; the all-failure mask $7$ is absent.  Since
$7=1+2+4$ would mean that all three candidates fail, its absence shows that
at least one candidate survives.  For the second implication, the prescribed
cells of $S_+\sqsubset Z_4$ locally determine $Z_3$ on
$[10,21]\times[29,30]$ to the two row words
\texttt{010000010100} and \texttt{000101000001}.
One-step scans with this band and $S_+\sqsubset Z_2$ exclude both
nonreference candidates.  In the last two-step transfer, a mismatch monitor
finds no extension disagreeing with any of the $56$ target cells.
Combining the three implications proves Lemma~\ref{lem:pump}.

\phantomsection\label{proof:marker}
\subsubsection*{Finite-state proof of Lemma~\ref{lem:marker}}
Use $Z_2\xrightarrow{\Life}Z_1\xrightarrow{\Life}Z_0$ with the common
premises $U\sqsubset Z_0$, $R\sqsubset Z_1$, and $R\sqsubset Z_2$.
Nine stages of boundary-automaton calculations prove
$U\sqsubset Z_2$.  Stages 4--9 reconstruct six disjoint geometric regions of
$U$, proceeding from the centre toward the outer parts of the marker.  Each
stage may contain many boundary transfers.
Set $D_{\mathrm c}=\{28,29\}\times[-2,12]$.  The two central
$30$-bit patterns are $Z_1|_{D_{\mathrm c}}$ and $Z_2|_{D_{\mathrm c}}$;
their reference values are $c_1|_{D_{\mathrm c}}$ and
$c_0|_{D_{\mathrm c}}$, respectively.
Table~\ref{tab:marker-stages} gives the dependencies and conclusions;
all target regions are in $Z_2$ unless otherwise indicated.

\begin{table}[ht]
\centering
\caption{The nine-stage marker restoration.  Premises are additional to the
common marker and rail premises.  Auxiliary supports in $Z_1$ are proved
before use.}
\label{tab:marker-stages}
\begin{tabularx}{\textwidth}{@{}r>{\raggedright\arraybackslash}X>{\raggedright\arraybackslash}X@{}}
\toprule
stage & inputs and prior stages & conclusion\\
\midrule
1 & Common premises & Seven boundary patterns inconsistent with the common premises are excluded.\\
2 & Stage 1
  & $Z_1|_{D_{\mathrm c}}=c_1|_{D_{\mathrm c}}$.\\
3 & Stage 2 & $Z_2|_{D_{\mathrm c}}=c_0|_{D_{\mathrm c}}$.\\
4 & Stages 2 and 3 & $766$ centre cells.\\
5 & Stage 4; derived $14$-cell support in $Z_1$ & $55$ right-fringe cells.\\
6 & Stage 5; derived $181$-cell support in $Z_1$ & $210$ shoulder cells.\\
7 & Stage 6; derived $189$-cell support in $Z_1$ & $361$ lower-slope cells.\\
8 & Stage 7 & $88$ bottom-tip cells.\\
9 & Stage 8; derived $605$-cell support in $Z_1$ & $975$ top-cap nonrail cells.\\
\bottomrule
\end{tabularx}
\end{table}

At each stage, overlapping cuts carry restrictions of the same global
layer.  Every additional premise follows from the common premises and
previous stages.  In particular, each auxiliary $Z_1$ support is first
derived from $Z_1\xrightarrow{\Life}Z_0$ and earlier supports, and is
then used in $Z_2\xrightarrow{\Life}Z_1$.  The artifact specifies the
domains, fixed values and boundary conditions for each stage.

The rail and the six target regions in stages 4--9 partition
$\operatorname{dom}(U)$:
\[
 1504+766+55+210+361+88+975=3959.
\]
The artifact's composition-check routine reconstructs these cell/value sets
and verifies their pairwise disjointness, their union and each stage's
premises.  This proves Lemma~\ref{lem:marker}.

\subsection{SAT proofs}
\label{app:clausal}
\phantomsection\label{proof:sat}

\subsubsection*{SAT proofs of Fact~\ref{fact:kynnos} and
Lemmas~\ref{lem:marker} and \ref{lem:pump}}

This subsection independently proves the three displayed implications by
showing that a finite SAT formula for a counterexample is unsatisfiable.  The
encoding and reconstruction argument is common to all three claims; their
individual domains, premises, conclusions, and checked certificates are
identified explicitly below.

\paragraph{Finite domains and global realizations.}
Write a depth-$d$ history as
$Z_d\xrightarrow{\Life}\cdots\xrightarrow{\Life}Z_0$.
Choose a finite $D_0$ and let
\[
 D_i=D_0+[-i,i]^2\quad(0\leq i\leq d),\qquad
 Z_{i-1}(p)=\Life(Z_i)(p)\quad(p\in D_{i-1},\ 1\leq i\leq d).
\]
All fixed premises lie in their respective $D_i$, and the target lies in
$D_d$.  For every equation centred at $p\in D_{i-1}$, the centre and all
eight input neighbours lie in $D_i$.  Thus every global Life history
restricts to a finite assignment satisfying the encoded local equations.
Conversely, every finite assignment satisfying all these complete local
equations extends to a global history: complete the oldest layer $D_d$
arbitrarily to the plane and iterate Life forward.  Induction down the
layers shows that the resulting history agrees with the finite assignment on
every $D_i$.

For each claim, encode these equations and fixed premises, and add a single
clause requiring at least one target cell to differ from its prescribed
value.  The resulting formula is satisfiable exactly when the stated
implication over arbitrary Life configurations has a counterexample.
Table~\ref{tab:finite-domains} lists
the domains and constraints.  For the pump, $S$ is the certificate-coordinate
translate of $S_+$ described in Section~\ref{app:interfaces}.  Its domain
$D_0$ includes the seed and target supports; the prescribed values are
exactly those in the fixed-premises column.

\begin{table}[ht]
\centering
\caption{Domains, premises, and desired conclusions for the three formulas.
Here $P@i$ means that $P$ is prescribed in $Z_i$; the SAT formula negates the
displayed desired conclusion.}
\label{tab:finite-domains}
\begin{tabular}{@{}lclll@{}}
\toprule
claim & $d$ & $|D_0|,\ldots,|D_d|$ & fixed premises & desired conclusion\\
\midrule
Kynn\"os & 1 & $584,688$ & $K@0$ & $K@1$\\
marker & 2 & $3959,4396,4776$ & $U@0;\ R@1,\ R@2$ & $U@2$\\
pump & 4 & $112,160,216,280,352$ & $S@0,\ S@2,\ S@4$ & $\shift_{(0,4)}S@4$\\
\bottomrule
\end{tabular}
\end{table}

\paragraph{Encoding and reconstruction.}
The primary formulas use a totalizer encoding \cite{bailleux2003totalizer}
of the local Life rule.  A binary-adder encoding provides a structurally
different local semantics cross-check.  To check the semantics with auxiliary
variables, for every
assignment to the centre and its eight neighbours fix the correct output
and require satisfiability over the auxiliary variables, then fix the
opposite output and require unsatisfiability.  The verification artifact
performs these tests for both encodings and for all four coordinate-parity
classes distinguished by the implementation:
$512\cdot2\cdot4=4096$ inputs and $8192$ satisfiability queries.  All pass.

The shared auxiliary variables represent threshold predicates for fixed
sets of input cells.  Each merge is defined by clauses in both directions,
and a threshold variable is shared only for the same set of input cells.
Thus any assignment satisfying the Life equations extends consistently to the
auxiliaries by assigning each threshold variable the truth value determined
by the actual neighbour count.  Conversely, induction through the merge gates
fixes every threshold used by the encoding to its uniquely determined truth
value, so
every satisfying assignment of the assembled formula obeys each local Life
equation.  Sharing
these definitions therefore preserves the equivalence on the whole patch.

The artifact reconstructs the three primary totalizer CNF formulas from the
stated patterns, layer domains, Life equations, fixed premises, and negated
targets; the resulting files are byte-for-byte identical to the supplied
files.  Separately, the binary-adder route checks local encoding semantics.
A third, cell-only encoding, written separately from the project encoder,
reconstructs the same three implications: each forbidden assignment of a
centre, its eight neighbours, and its output contributes a direct clause,
simplified using the fixed premises.  Its supplied LRAT certificates are checked by the formally verified
checker \texttt{cake\_lpr}~\cite{tan2021cakelpr}.

\paragraph{Certificate inventory.}
Table~\ref{tab:certificate-inventory} lists the three primary totalizer
formulas.
The clause counts include the final target-disagreement clause.
DRAT-trim~\cite{wetzler2014drat} checks the DRUP proofs of these formulas.
Their LRAT conversions~\cite{cruz-filipe2017lrat} are checked by
\texttt{lrat-check} from the same DRAT-trim revision and by the formally
verified checker \texttt{cake\_lpr}~\cite{tan2021cakelpr}.
For the marker and pump,
unsatisfiability proves Lemmas~\ref{lem:marker} and \ref{lem:pump} with
the premises listed in Table~\ref{tab:finite-domains}.  The first row recertifies the
same-position implication in Fact~\ref{fact:kynnos}.

\newcommand{\CertRow}[5]{#2 & #3 & #4 & #5\\}
\begin{table}[ht]
\centering
\caption{SAT certificates for the totalizer encoding.}
\label{tab:certificate-inventory}
\begin{tabularx}{\textwidth}{@{}>{\raggedright\arraybackslash}p{0.18\textwidth}>{\raggedright\arraybackslash}Xrr@{}}
\toprule
claim & implication & variables & clauses\\
\midrule
\CertRow{kynnos-core-self-force}{Kynn\"os}{$K@0$ forces $K@1$}{8,956}{37,093}
\CertRow{cap-spliced-marker-preimage-forces-marker}{marker}{$U@0$ and $R@1,\ R@2$ force $U@2$}{122,940}{529,034}
\CertRow{top-pump}{pump}{$S@0,\ S@2,\ S@4$ force $\shift_{(0,4)}S@4$}{11,456}{48,985}
\bottomrule
\end{tabularx}
\end{table}

The checked unsatisfiability of the three reconstructed formulas completes
the SAT proofs of Fact~\ref{fact:kynnos}, Lemma~\ref{lem:marker}, and
Lemma~\ref{lem:pump}.

%% file: main.bbl
\begin{thebibliography}{10}

\bibitem{gardner1970life}
Martin Gardner.
\newblock Mathematical games: The fantastic combinations of {John Conway's} new
  solitaire game ``{Life}''.
\newblock {\em Scientific American}, 223(4):120--123, 1970.

\bibitem{johnston2022conway}
Nathaniel Johnston and Dave Greene.
\newblock {\em Conway's Game of Life: Mathematics and Construction}.
\newblock Self-published, 2022.

\bibitem{rendell2002turing}
Paul Rendell.
\newblock Turing universality of the {Game of Life}.
\newblock In Andrew Adamatzky, editor, {\em Collision-Based Computing}, pages
  513--539. Springer, 2002.

\bibitem{salo-torma2022oracles}
Ville Salo and Ilkka T{\"o}rm{\"a}.
\newblock What can oracles teach us about the ultimate fate of {Life}?
\newblock In Miko{\l}aj Boja{\'n}czyk, Emanuela Merelli, and David~P. Woodruff,
  editors, {\em 49th International Colloquium on Automata, Languages, and
  Programming (ICALP 2022)}, volume 229 of {\em Leibniz International
  Proceedings in Informatics (LIPIcs)}, pages 131:1--131:20. Schloss
  Dagstuhl--Leibniz-Zentrum f{\"u}r Informatik, 2022.

\bibitem{salo-torma2025preimages}
Ville Salo and Ilkka~A. T{\"o}rm{\"a}.
\newblock Structure and computability of preimages in the {Game of Life}.
\newblock {\em Theoretical Computer Science}, 1042:115237, 2025.

\bibitem{ceccherini2010cellular}
Tullio Ceccherini-Silberstein and Michel Coornaert.
\newblock {\em Cellular Automata and Groups}.
\newblock Springer Monographs in Mathematics. Springer, Berlin, 2010.

\bibitem{kurka2003topological}
Petr K{\r u}rka.
\newblock {\em Topological and Symbolic Dynamics}, volume~11 of {\em Cours
  Sp\'ecialis\'es}.
\newblock Soci\'et\'e Math\'ematique de France, Paris, 2003.

\bibitem{lind-marcus1995}
Douglas Lind and Brian Marcus.
\newblock {\em An Introduction to Symbolic Dynamics and Coding}.
\newblock Cambridge University Press, 1995.

\bibitem{hu1987limit}
Lyman~P. Hurd.
\newblock Formal language characterizations of cellular automaton limit sets.
\newblock {\em Complex Systems}, 1(1):69--80, 1987.

\bibitem{culik1989limit}
Karel Culik, II, Jan Pachl, and Sheng Yu.
\newblock On the limit sets of cellular automata.
\newblock {\em SIAM Journal on Computing}, 18(4):831--842, 1989.

\bibitem{maass1995sofic}
Alejandro Maass.
\newblock On the sofic limit sets of cellular automata.
\newblock {\em Ergodic Theory and Dynamical Systems}, 15(4):663--684, 1995.

\bibitem{petersen1983ergodic}
Karl Petersen.
\newblock {\em Ergodic Theory}, volume~2 of {\em Cambridge Studies in Advanced
  Mathematics}.
\newblock Cambridge University Press, 1983.

\bibitem{bailleux2003totalizer}
Olivier Bailleux and Yacine Boufkhad.
\newblock Efficient {CNF} encoding of boolean cardinality constraints.
\newblock In Francesca Rossi, editor, {\em Principles and Practice of
  Constraint Programming -- CP 2003}, volume 2833 of {\em Lecture Notes in
  Computer Science}, pages 108--122. Springer, 2003.

\bibitem{tan2021cakelpr}
Yong~Kiam Tan, Marijn J.~H. Heule, and Magnus~O. Myreen.
\newblock {cake\_lpr}: Verified propagation redundancy checking in {CakeML}.
\newblock In Jan~Friso Groote and Kim~Guldstrand Larsen, editors, {\em Tools
  and Algorithms for the Construction and Analysis of Systems -- TACAS 2021},
  volume 12652 of {\em Lecture Notes in Computer Science}, pages 223--241.
  Springer, 2021.

\bibitem{wetzler2014drat}
Nathan Wetzler, Marijn J.~H. Heule, and Warren~A. Hunt, Jr.
\newblock {DRAT-trim}: Efficient checking and trimming using expressive clausal
  proofs.
\newblock In Carsten Sinz and Uwe Egly, editors, {\em Theory and Applications
  of Satisfiability Testing -- SAT 2014}, volume 8561 of {\em Lecture Notes in
  Computer Science}, pages 422--429. Springer, 2014.

\bibitem{cruz-filipe2017lrat}
Lu{\'i}s Cruz-Filipe, Marijn J.~H. Heule, Warren~A. Hunt, Jr., Matt Kaufmann,
  and Peter Schneider-Kamp.
\newblock Efficient certified {RAT} verification.
\newblock In Leonardo de~Moura, editor, {\em Automated Deduction -- CADE 26},
  volume 10395 of {\em Lecture Notes in Computer Science}, pages 220--236.
  Springer, 2017.

\end{thebibliography}
